\documentclass[sigconf]{aamas} 

\usepackage{balance} 

\usepackage{subcaption} 
\usepackage{url}

\usepackage{xspace}
\usepackage{stmaryrd}

\usepackage{color}

\def\green#1{\textcolor{green}{#1}}

\def\green#1{\textcolor[rgb]{0,0.6,0}{#1}}

\def\ignore#1{\iffalse #1 \fi}
\def\ignore#1{}
\def\memo#1{\iffalse #1 \fi}

\def\To{\shortrightarrow}

\usepackage[normalem]{ulem}
\newcommand\erase{\bgroup\markoverwith{\textcolor{red}{\rule[.5ex]{2pt}{0.4pt}}}\ULon}

\def\subst#1#2{\green{#1}}

\def\remove#1{\subst{}{#1}}

\def\Next{{\it next}}
\def\current{{\it c}}
\def\facilitator{{\it Fcl}}
\def\dst{{\it dst}}

\def\move{{\mathit{move}}}

\theoremstyle{plain}

\newtheorem{proposition}{Proposition}

\title[Distributed Planning with Asynchronous Execution]{Distributed
  Planning with Asynchronous Execution with Local Navigation for Multi-agent Pickup
  and Delivery Problem}

\author{Yuki~Miyashita}
\affiliation{
  \institution{Shimizu Corporation}
  \city{Tokyo}
  \country{Japan}}
\email{y-miyashita@shimz.co.jp}

\author{Tomoki~Yamauchi}
\affiliation{
  \institution{Waseda University}
  \city{Tokyo}
  \country{Japan}}
\email{t.yamauchi@isl.cs.waseda.ac.jp}

\author{Toshiharu~Sugawara}
\affiliation{
  \institution{Waseda University}
  \city{Tokyo}
  \country{Japan}}
\email{sugawara@isl.cs.waseda.ac.jp}

\begin{abstract}
 We propose a distributed planning method with asynchronous execution
 for {\em multi-agent pickup and delivery} (MAPD) problems for
 environments with occasional delays in agents' activities and
 flexible endpoints. MAPD is a crucial problem framework with many
 applications; however, most existing studies assume ideal agent
 behaviors and environments, such as a fixed speed of agents,
 synchronized movements, and a well-designed environment with many
 short detours for multiple agents to perform tasks easily. However,
 such an environment is often infeasible; for example, the moving
 speed of agents may be affected by weather and floor conditions and
 is often prone to delays. The proposed method can relax some
 infeasible conditions to apply MAPD in more realistic environments by
 allowing fluctuated speed in agents' actions and flexible working
 locations (endpoints). Our experiments showed that our method enables agents to
 perform MAPD in such an environment efficiently, compared to the
 baseline methods. We also analyzed the behaviors of agents using our
 method and discuss the limitations. 
\end{abstract}


\newcommand{\BibTeX}{\rm B\kern-.05em{\sc i\kern-.025em b}\kern-.08em\TeX}

\begin{document}


\pagestyle{fancy}
\fancyhead{}


\maketitle 


\section{Introduction}
{\em Multi-agent path-finding} (MAPF), in which multiple agents move
to their destinations by avoiding collisions, is an important abstract
problem that arises in many applications, such as robotics and
games~\cite{salzman2020research,Jennings1997,Michal2006}. For example,
in an automated warehouse, the robots (carrier agents) move to the
pickup locations, load materials, and deliver them to their respective
unloading locations, by repeatedly assigning new tasks to agents
consecutively. This type of application in which MAPF problems are
iteratively solved is formulated as a {\em multi-agent
  pickup-and-delivery} (MAPD) problem~\cite{MaAAMAS2017}. Therefore,
the objective of MAPD is that agents repeatedly move to their
respective {\em task endpoints}, which are pickup and delivery
locations~\cite{MaAAMAS2017}, without collisions. 
\par

However, complex and restrictive environments may reduce efficiency
and hence increase the chances of collisions. Furthermore,
uncertainties often exist in practice. For example, The
  occasional delays in the movements
of some agents will affect the plans of others. Consequently, many agents cannot move according to
their plans as previously determined. In such a case, these agents
have to replan their paths to avoid collisions. However, replanning
requires considerable computational cost with an increasing number of
agents when using centralized methods or when assuming synchronized
movement with decentralized control. Meanwhile, we envisioned an
application for robotic material transfer in a construction site and a
disaster area and such uncertainty is not preventable. Thus, we
focused on a distributed 
planning-and-executing method wherein agents generate their plans
individually and move to their destinations interacting only with the
local agents involved to address the negative effects of delays and
resource conflicts. Although distributed planning for MAPD has some
issues to be addressed, such as completeness and avoidance of
live-/dead-lock, we think that a distributed planning method is
preferable because of their desirable characteristics, such as
robustness, scalability, and reconfigurability~\cite{robotYan2013}. 
\par

Although there are several studies on planning and execution in MAPF
and MAPD, most studies assumed grid-like environments, such as an
automated warehouse wherein (1) there are many endpoints more than
agents and (2) the agents move at a constant speed and move
synchronously~\cite{MaAAMAS2017}. Moreover, these conventional methods
require a {\em well-formed infrastructure conditions} (WFI
conditions)~\cite{Michal2015ICAPS}. For example, in the {\em holding
  task endpoint} (HTE)~\cite{MaAAMAS2017}, any path connecting a pair
of endpoints does not traverse any other endpoints. These types of
requirements can be easily satisfied in a simple grid-like
environment, but are not feasible in our target applications. For
example, at a construction site, heavy-duty robots transport heavy
building materials weighing 500 kg and above between storage
locations, specific work locations, and/or elevators that carry them
to other floors, during the night for the next day's work; the number
of these endpoints (i.e., storage and working locations) is not
large. Therefore, several agents are likely to move to a few specific
endpoints as their destinations simultaneously, but should not collide
with each other. Furthermore, (1) owing to a variety of working
locations, the WFI conditions cannot necessarily be met, and (2)
agents cannot often move at a constant speed because of various
reasons, such as temperature, humidity, wet floor, slopes/small steps,
and sensor errors.
\par

Our contribution is twofold. First, we introduced a problem of {\em MAPD with fluctuated
  movement speed} (MAPDFS), which is an abstraction of our target
application that involves material transport tasks by multiple robots
at a construction site, considering the characteristics discussed
above, such as a smaller number of endpoints\remove{ (it may be less than
agents)}, and uncertainty in agent's movement speed.
\par

Second, we presented a novel distributed planning
with an asynchronous execution method for MAPD problems. Its features
are low-cost planning without considering the plans of other agents
and the adjacent movements of multiple agents, based on possible
conflict detection in an environment described by the graph consisting
of several bi-connected components with small tree-structured
areas. We introduced into MAPD two types of agents: {\em carrier
  agents}, simply called agents, that carry materials while planning
paths with ignoring other agents, and {\em node agents} that manage
the corresponding nodes in the graph and detect the possibility of
conflicts by communicating only with the neighboring nodes. After a
carrier agent plans a path, it communicates with the node agent
managing the current node, and then the node agent requests the next
node agent to reserve the node. Depending on the response, the current
node agent suggests to the carrier agent that it can move to the next
node, should move to another node (i.e., taking a detour) or should
wait for a while at the current node. Then, it moves/waits
asynchronously without considering other agents. Moreover, when it
moves to another node that is not in the current plan, it replans
another path, ignoring other agents.
\par

The basic idea behind the ignorance of other agents is the fact that
detailed planning with rigorous execution will easily become
unavailable in a real environment that may have
uncertainty. Conversely, to reduce collisions, we introduced the
orientation to the environment such that its main area can be observed
as a {\em strongly connected graph}. Subsequently, the direct edges
are used to navigate the agents based on the directions and eliminate
head-on collisions and dead-/livelock, such as a roundabout. This may
force agents to take detours; however, it simplifies the independent
planning process, makes it easier to detect potential collisions,
allows agents to delay/stop moving, and relaxes the WFI conditions
with a small number of endpoints. These features are indispensable in
real-world applications. We then experimentally evaluated the
performance of the proposed method by comparing it with two baseline
methods, {\em rolling-horizon collision
  resolution} (RHCR)~\cite{Li2021AAAI} and HTE~ \cite{MaAAMAS2017}. We
found that our proposed 
method outperformed the baseline methods in the environment that meets
the conditions they require. We also showed that agents with our
method can complete all tasks without collisions, despite occasional
delays in agent speed and flexibility in locating the endpoints.


\section{Related Work}
Although finding collision-free paths is a basic requirement for MAPF/MAPD, it is known as an NP-hard problem~\cite{JingjinAAAI2013}, which involves high computational costs. Therefore, several studies attempted to find sub-optional solutions by relaxing problems~\cite{barer2014suboptimal,BoyarskiIJCAI2015,FelnerICAPS2018}. 
\par

MAPF is often integrated into MAPD by iterating path-finding using
centralized
planners~\cite{Nguyen2017IJCAI,cohen2015feasibility,huang2021AAMAS,liu2019AAMAS}. For
example, Nguyen et al.~\cite{Nguyen2017IJCAI} solved task assignment
and path-finding problems using answer set programming. Cohen, Uras,
and Koenig~\cite{cohen2015feasibility} extended {\em conflict-based
  search} (CBS)~\cite{sharon2015conflict} by adding a set of edges
with the user-specified orientation, called {\em highways}, that
provide collision-avoidance guidance by eliminating agents travel in
the opposite direction. The {\em highways} also contribute to
decreasing the runtimes and solution costs of the MAPF solver for
autonomous robots in a warehouse. However, these centralized planners
do not scale because of high computational and communication costs
with an increasing number of agents. Our method can be considered a
combination of {\em highways} and the distributed planning and
execution because we also oriented edges (algorithmically) in
environments to locally 
navigate agents. 
\par

Other sub-optimal centralized planning methods~\cite{vsvancara2019AAAI,Li2021AAAI,WiktorIROS2014} decompose iterative MAPFs into sequential path-planning problems in which a solver replans all paths at every timestep. Because these methods have computationally high costs, some studies attempted to improve scalability. \v{S}vancara et al.~\cite{vsvancara2019AAAI} considered an online MAPF in which new agent goals may be added dynamically and disappear at their goals. They then proposed the online independence detection algorithm for generating plans to reduce disruptions of existing plans. Li et al.~\cite{Li2021AAAI} proposed RHCR in which agents find conflicts that occur in the $w$-timestep windows and replan paths every $h$ timesteps. However, if the environment is dense and crowded, many paths for other agents have to be checked to generate a collision-free path, resulting in high computational costs.
\par

The decentralized planning of collision-free paths has also been
studied~\cite{MaAAMAS2017,ma2019lifelong,WiktorIROS2014,Okumura2019IJCAI,Yamauchi2022AAMAS}. For
example, Ma et al.~\cite{MaAAMAS2017} proposed HTE, which is a
prioritized path planning method wherein agents plan their paths
consecutively and the new path is required \remove{not} to be
collision-free with existing paths of other agents. The {\em priority
  inheritance with back-tracking} (PIBT)~\cite{Okumura2019IJCAI}
focuses on the adjacent movements of multiple agents based on
prioritized planning in a short time window. These methods require
environmental conditions to guarantee the completeness of MAPD
instances. HTE requires the WFI condition and PIBT requires that the
environment is a bi-connected graph. Although these environmental
conditions apply to some applications (e.g., MAPD in an automated
warehouse), the environments of other applications (e.g., construction
sites and some logistics sorting
centers~\cite{Wan2018ICCARV}) cannot
satisfy these conditions. Furthermore, agents are forced to
synchronize movements in decentralized approaches such as these and
thus cannot handle fluctuations in moving speed.
\par

Meanwhile, we assumed that distributed approaches are more suitable in our applications because agents can move asynchronously and improve scalability and robustness. However, the methods from this approach are sometimes incomplete. For example, Wilt and Botea~\cite{Wilt2014ICAPS} proposed a spatially distributed planner in which each controller agent manages a subarea and communicates with the adjacent controller confirming the transfer of a mobile unit to another subarea. However, this method is partially centralized to make it complete. Thus, its environmental condition is somewhat inflexible, and agents cannot move asynchronously. Miyashita, Yamauchi, and Sugawara~\cite{Miyashita2022COMPSAC} proposed a distributed planning that enables agents to move asynchronously. Their method introduces two types of agents such as ours, {\em carrier agents}, which carry materials, and {\em node agents}, which manage the resource conflicts. However, their method is significantly complicated to implement. Our method takes a similar approach but is much simpler and low-cost, and agents can move effectively even in a crowded area.


\def\AgentSet{I}
\def\TaskSet{\mathcal{T}}
\def\task{\tau}

\def\park{{\it pk}}
\def\move{{\it mv}}
\def\loadunload{{\it lu}}
\def\main{{\it main}}
\def\marginal{{\it mar}}
\def\tree{{\it tree}}
\def\Root{{\it root}}
\def\abs#1{{\vert #1\vert}}

\section{Problem formulation}
\subsection{MAPD}
In the MAPD problem, {\em carrier agents}, or simply agents load
materials at specified locations, carry them to specified
destinations, and unload them. Let $\AgentSet =\{1, \cdots, n\}$ be
the set of agents and $\TaskSet =\{\task_1, \cdots,
  \task_{\abs{\TaskSet}}\}$ be 
a set of tasks that should be completed by the agents. An environment
is described by a {\em connected} graph $G=(V,E)$ that can be embedded
into a two-dimensional space, where $V = \{v_1, \cdots, v_{\abs{V}}\}$ is
a set of nodes corresponding to locations and $E$ is a set of edges
corresponding to connections between two neighboring nodes. Our graph
is constituted at a combination of indirect and direct edges, and
hence the edge connecting $v_j$ and $v_k$ is denoted as either
$(v_j,v_k) \in E$ for the indirect edge or $(v_j \To v_k) \in E$ for
the direct edge. Therefore, $(v_j \To v_k)$ indicates the path through
which agents can move only from $v_j$ to $v_k$, while agents can move
to both directions between $v_j$ and $v_k$ on $(v_j,v_k)$. Note that
we assume that agent $i$ cannot stop on an edge\remove{ at the same
  time}. Two agents collide only when they exist on a single node or
traverse the same edge in opposite directions; this implies that a
collision does not occur on the direct edge. 
\par

Task $\task_k\in\TaskSet$ is specified by a tuple $\task_k=(v_k^l,
v_k^u, \mu_k)$, where $\mu_k$ is the material to carry, $v_k^l$ is the
pickup node where to load $\mu_k$, and $v_k^u$ is the destination node
where to unload $\mu_k$. The nodes where agents load and unload
materials are called {\em task endpoints}. We introduce discrete time
$t\geq 0$, whose unit is {\em timestep}. When agent $i\in\AgentSet$ is
allocated task $\task_k\in \TaskSet$ at time $t$, it individually
plans to generate a path from \remove{the current} node
$v_\current^i(t)$
  to $v_k^l$ to
load $\mu_k$ and a path from $v_k^l$ to $v_k^u$ to unload it, where
$v_\current^i(t)$ is the node on which $i$ is at
  $t$. Then,\remove{agent} $i$ moves to 
$v_k^l$ and $v_k^u$ along the paths and load/unload $\mu_k$. If a
collision is possible, $i$ has to modify the path to avoid it. 
\par

Initially ($t=0$), agent $i$ starts from its own parking node,
$v_i^\park$ ($=v_\current^i(0)$), and then begins to execute the given tasks
in $\TaskSet$ with other agents. An agent can move to one of
neighboring nodes in $T_\move$ timesteps if possible, but may need
longer time $T_\loadunload$ ($\geq T_\move$) for loading/unloading. It
can also wait at a node for any timesteps. We may assume
fluctuations in the agent's speed; that is, it occasionally takes
longer than $T_\move$ to move to a neighboring node. When $i$ has
completed the current task, one element in $\TaskSet$ is allocated to
it. When $\TaskSet=\varnothing$, $i$ returns to its parking
node. We define that an {\em endpoint} is a task endpoint or a parking
node.
\par

\subsection{Required Environmental Conditions}

We assume our environment $G=(V,E)$ consists of some bi-connected
components with small trees. We \remove{then} define it more
formally.
Let $G$ be a undirected graph and $G'=(V',E')$ be
a subgraph of $G$ associated with $V'\subset V$, that is,
$(v_1,v_2)\in E'$ iff $v_1, v_2\in V'$ and $(v_1, v_2)\in E$.
\begin{definition} (Bi-connected component)
  \begin{itemize}
    \item[(a)] $G'$ is a {\em bi-connected subgraph} of $G$ iff, for $\forall v_1, v_2\in V'$, $G'$ contains a cycle connected $v_1$ and $v_2$.
    \item[(b)] Bi-connected subgraph $G'$ of $G$ is a maximal if $\not\exists v\in V\setminus V'$ s.t. the subgraph of $G$ associated with $V'\cup\{v\}$ is the bi-connected subgraph.
    \item[(c)] A maximal bi-connected subgraph of $G$ is called a {\em bi-connected component}.
  \end{itemize}
\end{definition}
Let $G_1, \dots, G_K$ be all bi-connected components of $G$.
We denote $G_k=(V_k,E_k)$\remove{ and $K\geq 1$ is an
  integer}. We call $G_\main= G_1\cup 
\dots\cup G_K$ the {\em main area} of $G$, where we define the union
of graphs as $G_k\cup G_{k'}=(V_k\cup V_{k'},E_k\cup E_{k'})$. We
assume that our environment $G$ holds the following structural (SC)
and agent (AC) conditions: 
\begin{itemize}
  \item[SC1.] $G_\main$ is a connected graph.
  \item[SC2.] For $\forall v\in V\setminus V_\main$, $v$ is a node in
    one of $L$ \remove{($\geq 0$)} tree-structured subgraphs of $G$,
    $G_\tree^1, \dots, G_\tree^L$, and for $1\leq k\leq L$,
    $G_\tree^k\cap G_\main$ is a singleton whose element is the root
    node of $G_\tree^k$. 
  \item[SC3.] Parking nodes are end nodes in tree-structured subgraphs that have no task endpoints.
  \item[AC1.] Agents are at least two fewer than the nodes in the main area, $\abs{\AgentSet} \leq \abs{G_\main}-2$, but we suggest that agents be fewer than half of $\abs{G_\main}$ for efficiency. This topic will be discussed later. 
\end{itemize}
We denote the set of all root nodes as $G_\Root=(G_\tree^1\cup\dots\cup G_\tree^L)\cap G_\main$. By defining {\em marginal zone} as $G_\marginal=G_\tree^1\cup\dots\cup G_\tree^L\setminus G_\Root$, $G=G_\main\cup G_\marginal$ is evidently a disjoint union. Then, the following proposition is true.
\par

\begin{proposition}
 Suppose that $G_\main=G_1\cup \dots\cup G_K$ is connected graph and $G_k$ is the bi-connected component. If $V_k\cap V_{k'}\not=\varnothing$, the $V_k\cap V_{k'}$ is a singleton.
\end{proposition}
\begin{proof}
 If $V_k\cap V_{k'}$ has two nodes, $V_k\cup V_{k'}$ is bi-connected; thus, $V_k$ and $V_{k'}$ are not bi-connected components, and this is a contradiction (the detailed proof is provided in Appendix).
\end{proof}

We can set a task endpoint to any nodes in $G$ subject to Condition
SC3 If an agent loads/unloads
a material on the nodes in the main area, the movement of other agents
may temporarily impeded. However, because working places are possible
anywhere on a construction site, we believe that it is inevitable and
agents should carry the required materials to such places. 
\par

\begin{figure}
  \begin{minipage}[b]{0.49\linewidth}
  \centering
  \includegraphics[width=0.85\hsize]{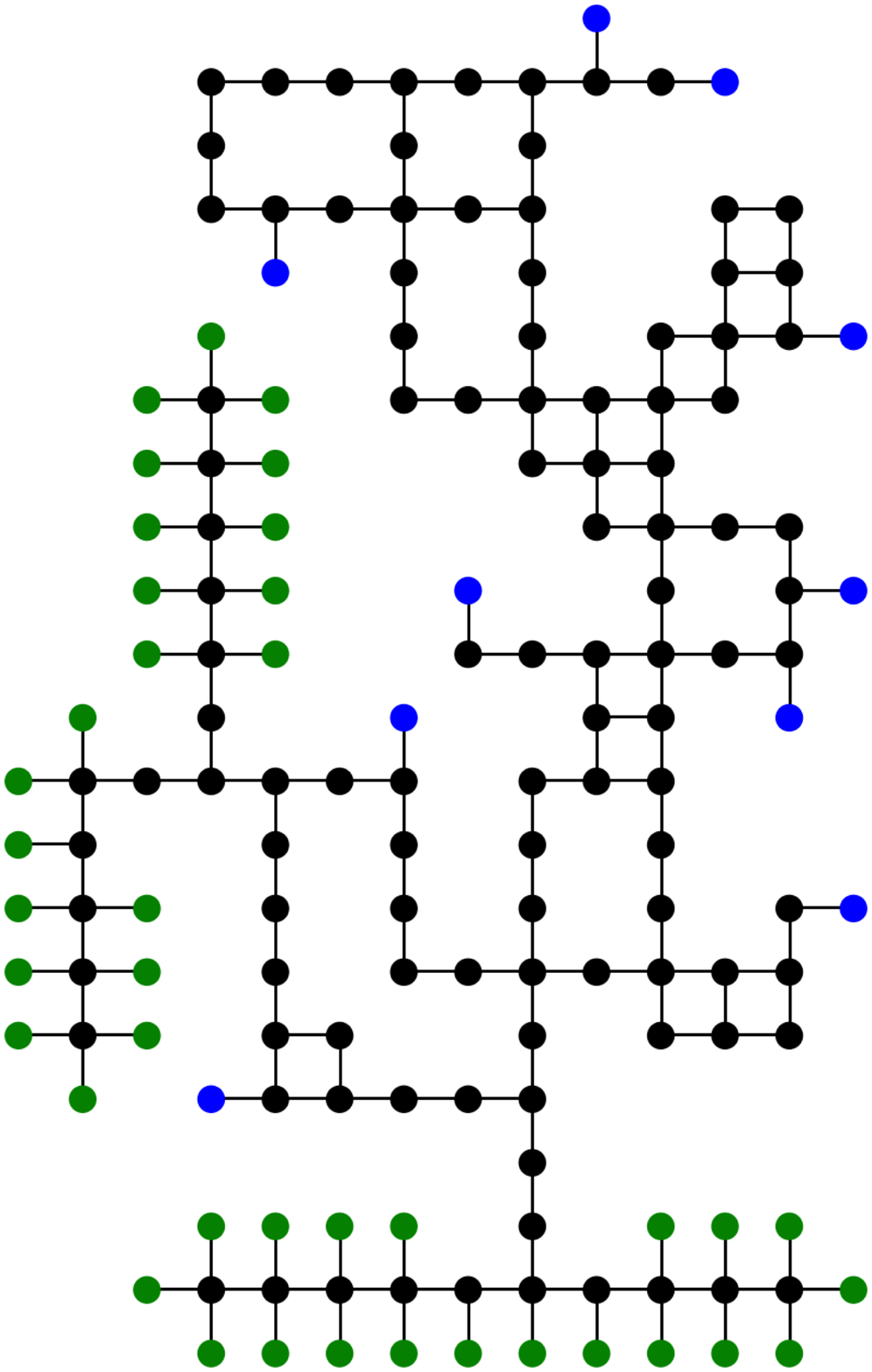}
  \subcaption{Environment 1 (Env.~1)}\label{fig:env1}
  \end{minipage}
    \begin{minipage}[b]{0.49\linewidth}
  \centering
  \includegraphics[width=0.85\hsize]{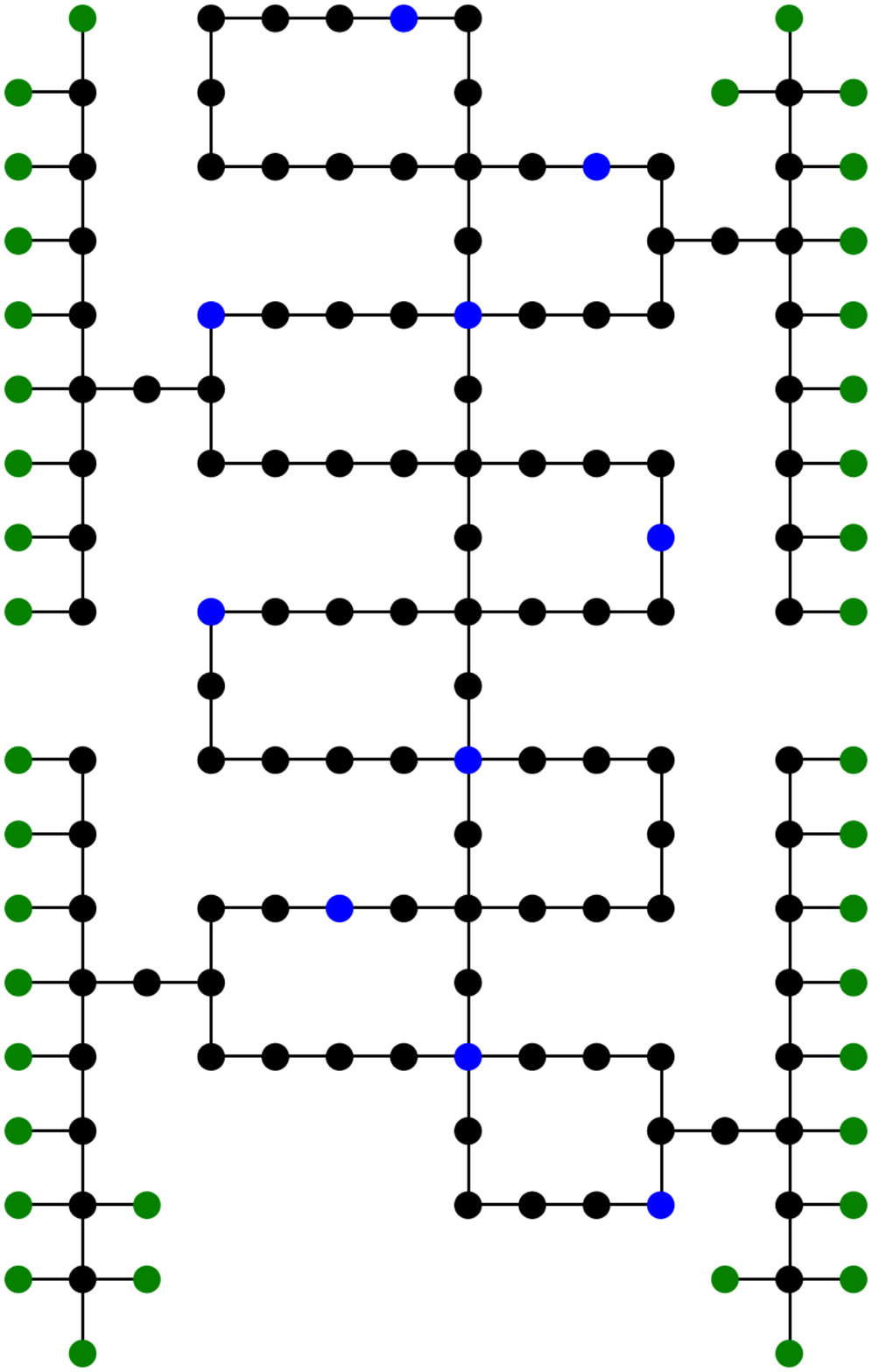}
  \subcaption{Environment 2 (Env.~2)}\label{fig:env2}
  \end{minipage}
   \caption{Example environments.}
    \label{fig:env}
\end{figure}

\begin{figure}
  \begin{minipage}{0.45\linewidth}
  \centering
  \includegraphics[width=1\hsize]{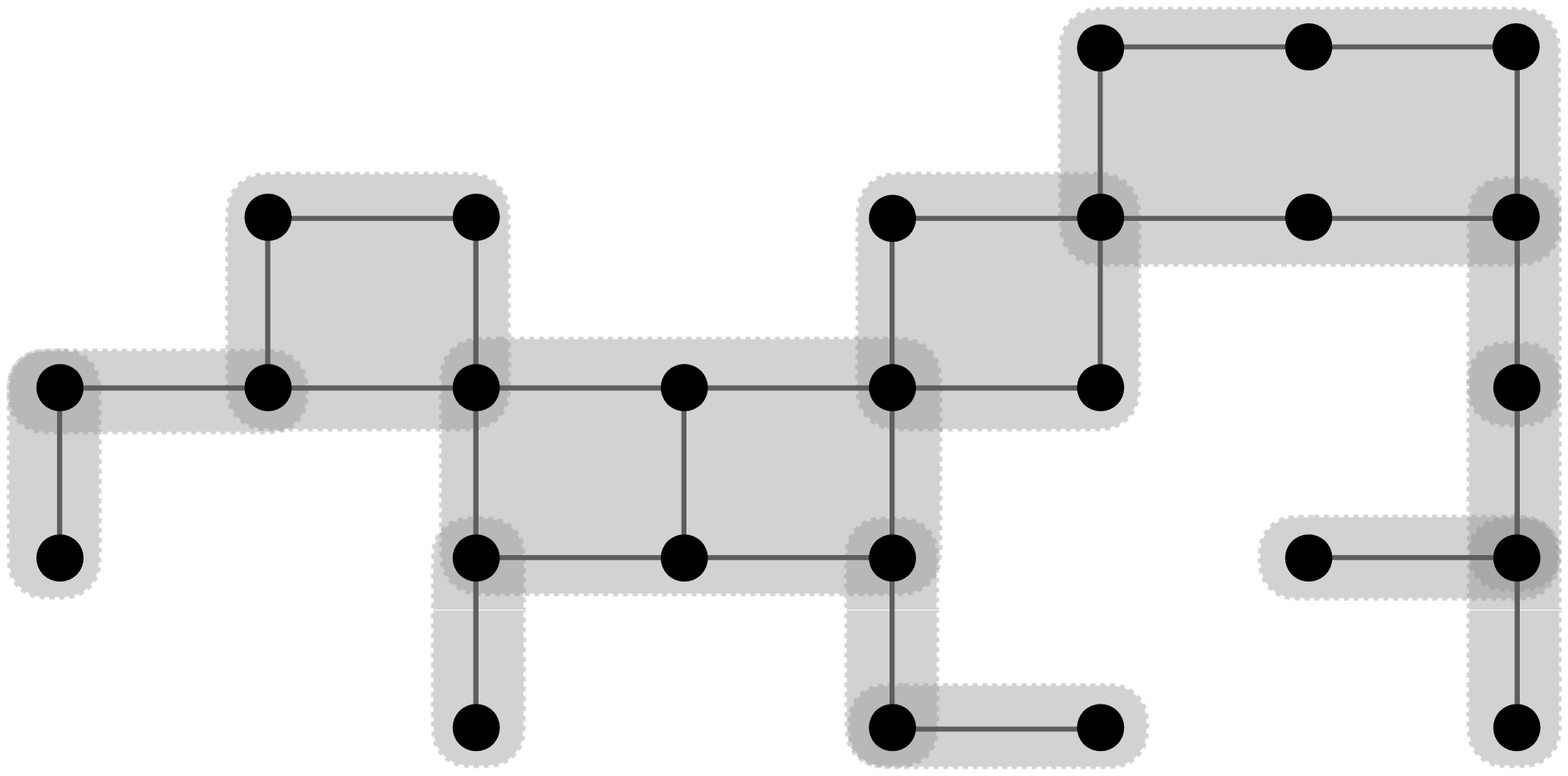}
  \subcaption{Graph structure.}
  \label{fig:bf_meet_graph}
  \end{minipage}
  \begin{minipage}{0.45\linewidth}
  \centering
  \includegraphics[width=1\hsize]{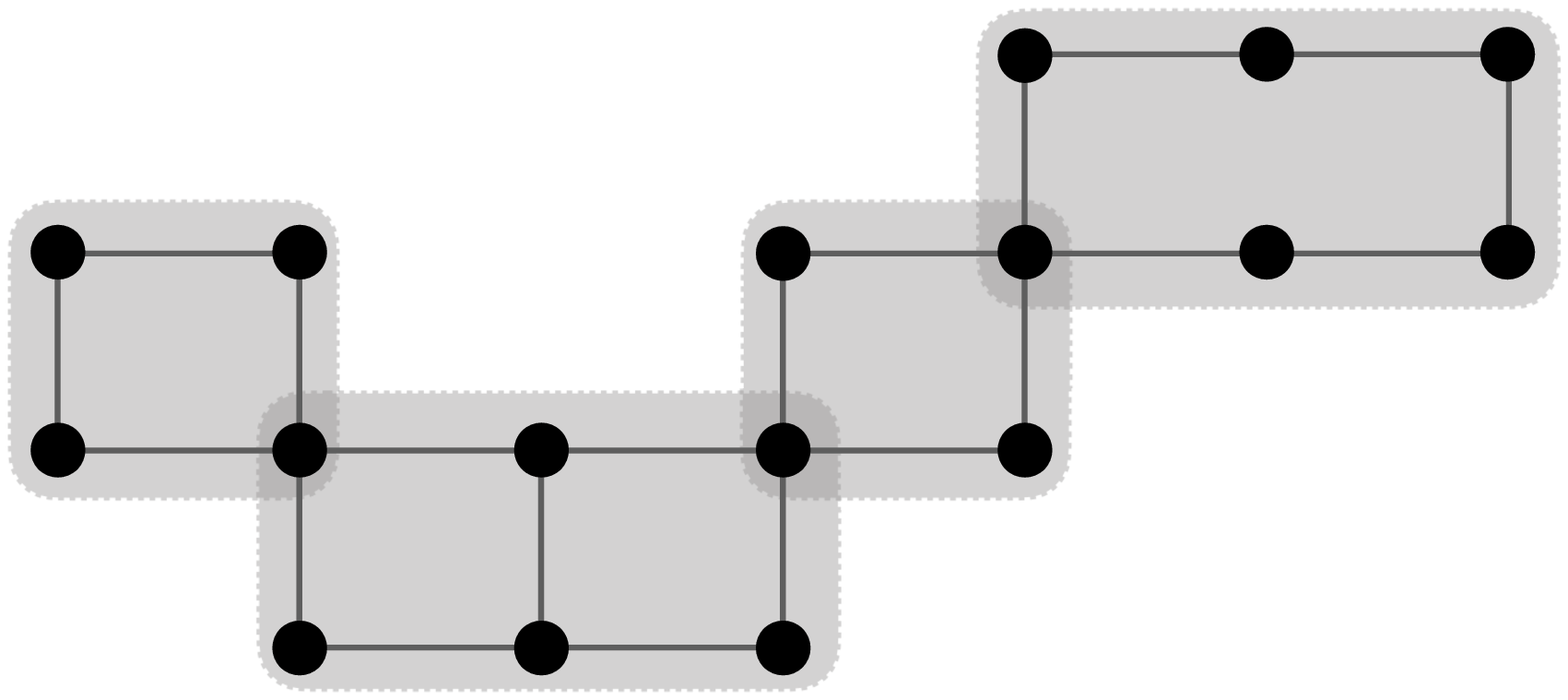}
  \subcaption{Main area.}
  \label{fig:af_meet_graph}
  \end{minipage}
  \caption{Example environment.}
  \label{fig:meet_graph}
\end{figure}

\begin{figure}
  \begin{minipage}[c]{0.45\linewidth}
  \centering
  \includegraphics[width=1\hsize]{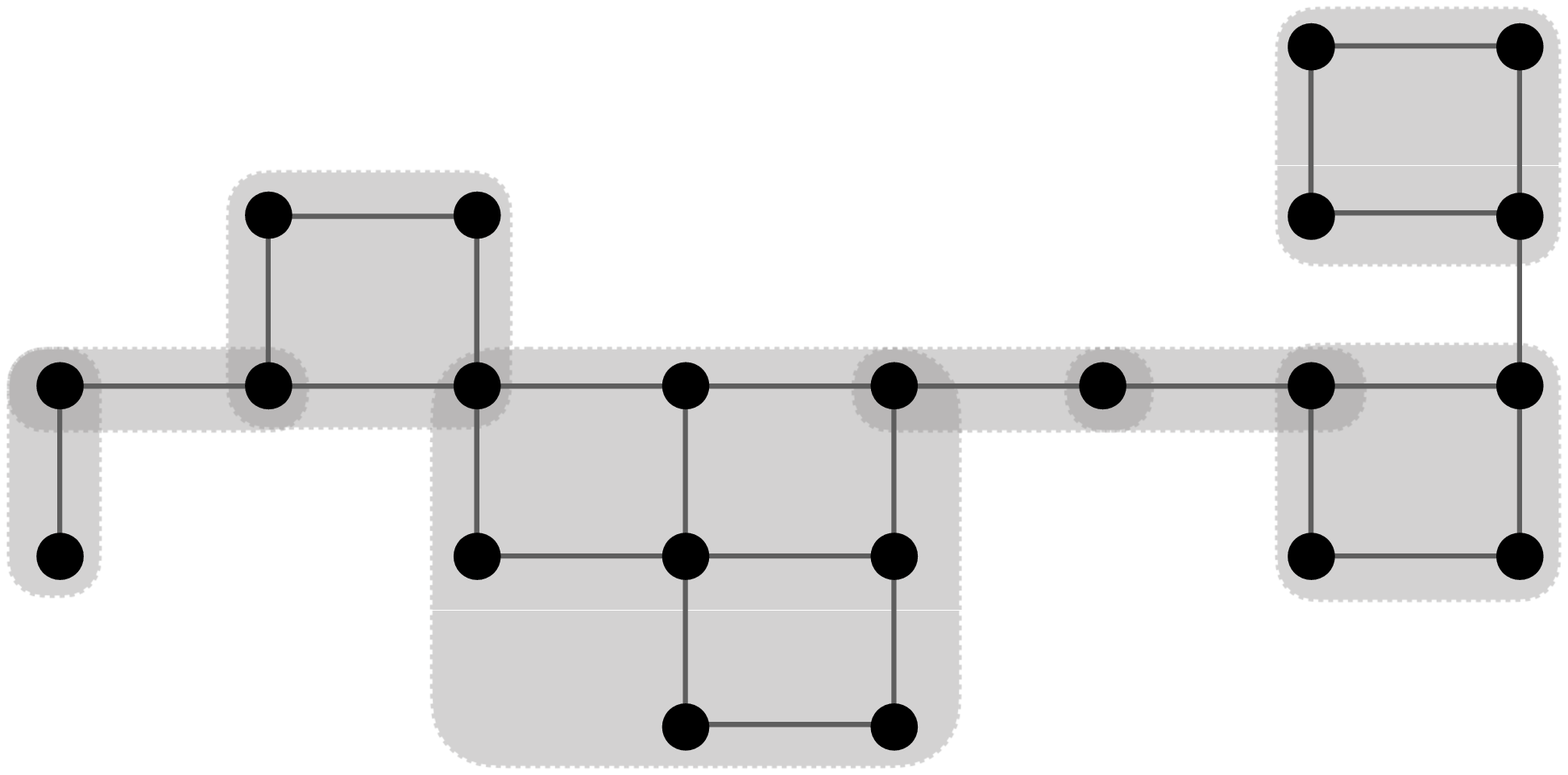}
  \subcaption{Graph}
  \label{fig::bf_not_meet_graph}
  \end{minipage}
  \begin{minipage}[c]{0.45\linewidth}
  \centering
  \includegraphics[width=1\hsize]{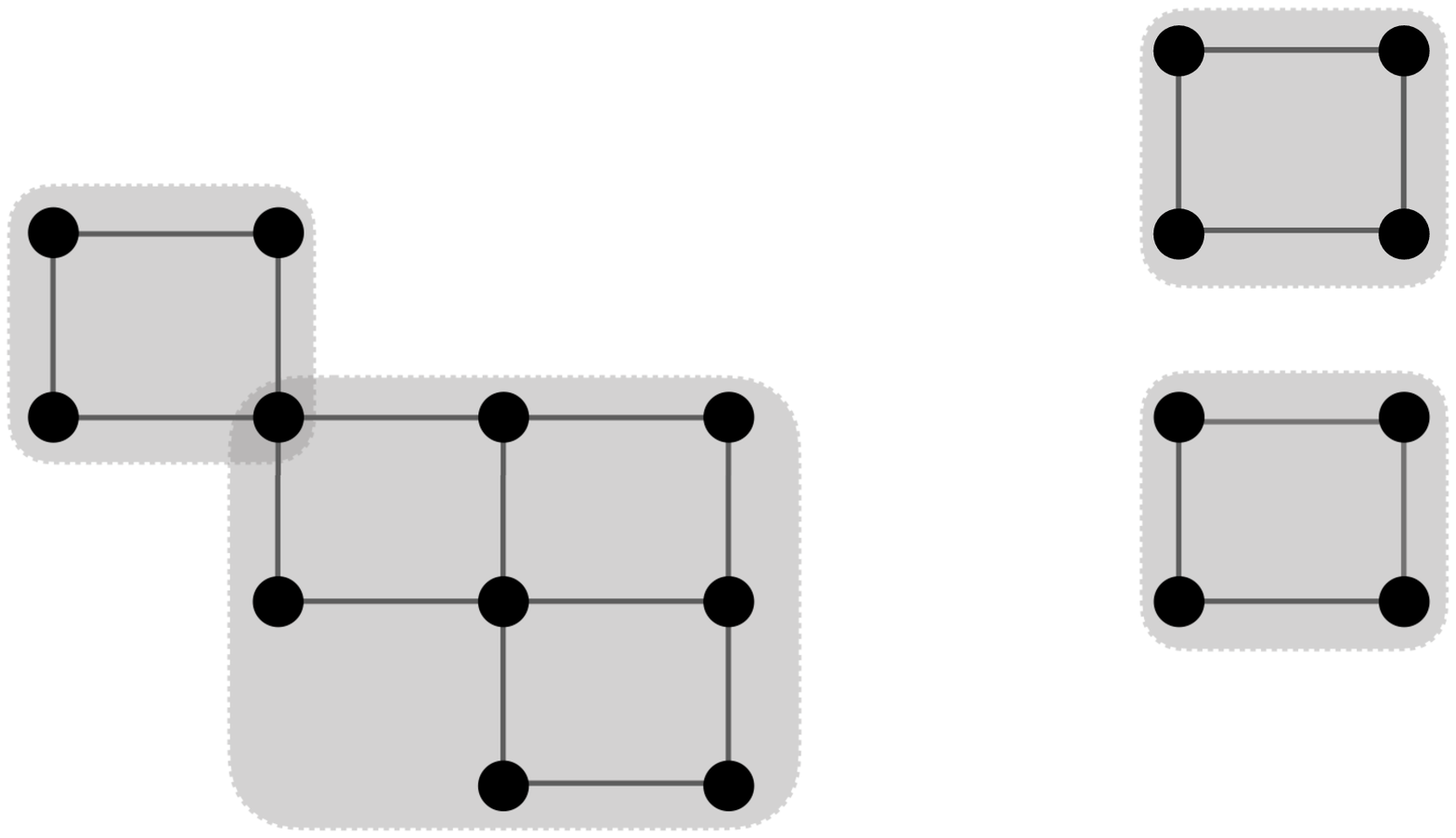}
  \subcaption{Main area is disconnected.}
    \label{fig:af_not_meet_graph}
  \end{minipage}
  \caption{Graph that does not hold the conditions.}
  \label{fig:not_meet_graph}
\end{figure}

The example of our environment is shown in Fig.~\ref{fig:env}, where (forty) green dots are parking nodes for individual agents and blue dots are task endpoints (i.e., loading and unloading locations). Another example is shown in Fig.~\ref{fig:meet_graph}; Fig.~\ref{fig:bf_meet_graph} shows the environment, and Fig.~\ref{fig:af_meet_graph} is its main area. Fig.~\ref{fig:not_meet_graph} shows the graph that does not meet our conditions: actually, the main area is not connected (Fig.~\ref{fig:af_not_meet_graph}), and some nodes in $G_\marginal$ are not nodes in trees. 
\par


\section{Proposed method}
One feature of the proposed algorithm is simple planning by ignoring the other agents' current plans. Instead, for $\forall v\in V_\main$, we introduce the {\em node agent}, which manages the resource of $v$ and detects the possibility of conflict, that is, collisions between agents. While a (carrier) agent communicates only with the {\em current node agent} that manages the resource of the current node, a node agent communicates with the neighboring node agents and checks the possibility of movement of the agent on itself. The node agent is also represented hereinafter by $v$. A carrier agent generates a path, which is a sequence of nodes to the destination independently and asynchronously, and moves to the next node based on the path while making sure that another agent does not stay at the next node or move toward it by asking the current node agent.
\par

Note that node agents are not necessarily located on the corresponding
nodes because they only manage the reservation of the corresponding
nodes. Thus, they can run on a single machine, on different
  servers in a cloud, or on intelligent sensors near the
  locations. The
only requirement is that the node agent should be able to communicate
with the node agents that manage the neighboring nodes and with the
carrier agent that is currently reserving the node.

\subsection{Orienting Graphs with Reachability}
Our basic idea is to introduce a {\em strong orientation} in the main
area to prevent crossing the same edge, particularly moving in
opposite directions along a long straight path consisting of multiple
edges. This may result in detours, but conversely allows consistent
structural direction for the flows of movements. Thus, agents can
avoid collisions with only local information and resource allocation,
address travel delays, or \remove{more extreme and} sudden stops in an
opportunistic manner. This also eliminates the need for costly
planning considering other agents' paths and negotiation
\remove{processes} to avoid collisions. 
\par

First, we \remove{briefly} introduce the basic concepts related to graph theory.

\begin{definition}
  A directed graph is a {\em strongly connected} iff any pair of nodes has paths in both directions between them.
\end{definition}
\begin{definition}
  An edge in an undirected connected graph is called a {\em bridge} iff the graph is not connected anymore if it is eliminated.
\end{definition}
Evidently, our main area $G_\main$, is bridgeless (or {\em 2-edge
  connected}) and \remove{strongly} connected. Then, the following
theorem is known as the {\em one-way street
  theorem}~\cite{Robbins1939}.
\begin{theorem}
  A bridgeless connected undirected graph can be made into a strongly
 connected graph by consistently orienting (and vice versa).
\end{theorem}

Several efficient algorithms (linear and $\log \abs{E}$) to orient a bridgeless connected undirected graph to make it strongly connected have been proposed~\cite{Atallah1984,Dijkstra76,Tarjan1972}. We orient $G_\main$ using one of these algorithms. Note that edges in $G_\marginal$ remain undirected, which means bi-directional edges that an agent travels in both directions. An example of the oriented environment of Fig.~\ref{fig:bf_meet_graph} is shown in Fig.~\ref{fig:directedGraph}.
\par

\begin{figure}
  \centering
  \includegraphics[width=0.55\hsize]{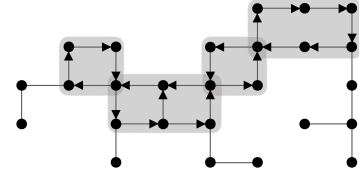}
  \caption{Example of orienting graphs.}
  \label{fig:directedGraph}
\end{figure}

\subsection{Behavior of Carrier Agents}

We assume that the environment $G$ has already been oriented, as
described in the previous section. When task $\tau_k=(v_p^k,
v_d^k,\mu^k)$, is allocated to carrier agent $i$, $i$ will move to its
load location $v_p^k$, and to its unload location $v_d^k$. Therefore,
$i$ sets the destination node, $v_\dst \in V$, to $v_p^k$ or $v_d^k$
in $\tau_k$, depending on the phase of the task progress and then
generates the shortest path $p_i$ (or appropriate path from another
perspective) from the current node $v_\current^i$, to $v_\dst$, using
a conventional method (e.g., $A^*$-search) in the (partly) directed
graph $G$. Herein, we define a path $p$ from node $v$ to node $v'$
\remove{in the 
proposed method} as the sequence of nodes $p=(v_0^p, v_1^p, \dots
v_\abs{p}^p)$, where $v_0^p= v$ and $v_k^p=v'$ and $v_{k-1}^p$ and $v_k^p$
are connected by edge $(v_{k-1}^p, v_k^p)$ or $(v_{k-1}^p\To
v_k^p)$. Note that, unlike other methods for MAPD, we can generate a
path by ignoring time information, that is, when agents arrive and
leave nodes. After generating path $p_i$, $i$ attempts moving to $v_\dst$
in line with $p_i$. 
\par

We denote the current node of agent $i$ at $t$ by
$v_\current^i(t)$. Agent $i$ has the {\em facilitator node agent} (or
simply, {\em facilitator}), $v_\facilitator^i(t)\in V_\main$, which is
identical to the current node $v_\facilitator^i(t)=v_\current^i(t)$ if
$v_\current^i(t)\in V_\main$, and if $i$ is in a tree (i.e.,
$v_\current^i(t)\remove{=v}\in G_\tree^k$), its facilitator node is set to the
root node, $v_\facilitator^i(t)\in G_\tree^k\cap G_\main$. If
$v_\current^i(t)$ and the next node $v_\Next^i \in p_i$ are not in
$V_\main$, $i$ moves to $v_\Next^i$ without
  confirmation. Otherwise (i.e., if 
$v_\current^i(t) \in V_\main$ or $v_\Next^i \in V_\main$), before $i$
at $v_\current^i(t)$ moves to next node $v_\Next^i \in p_i$, $i$ sends
a {\em request message} to node agent $v_\facilitator^i$ with $p_i$ to
reserve $v_\Next^i$. Then, it will receive its reply from
$v_\facilitator^i$. If it is an {\em acceptance message}, $i$ leaves
the current node for $v_\Next^i$ at $t+1$ and releases the reservation
for $v_\current^i(t)$. Note that it may take some time $T_\move \geq
1$ to reach $v_\Next^i$; however, we assume that
$v_\current^i(t+1)=v_\Next^i$ until $i$ leaves there. It also means
that their activities are asynchronous if $T_\move > 1$; that is, when
an agent starts leaving the current node, other agents may already be
in the middle of edges. After $i$ reaches the next node, $i$ attempts
to reserve the next node $v_\Next^i$ based on plan $p_i$.

\def\waitSOM{\textsc{wait}\xspace}
\def\detourSOM{\textsc{detour}\xspace}
\def\WaitSOM{\textsc{Wait}\xspace}
\def\DetourSOM{\textsc{Detour}\xspace}

Meanwhile, if $i$ receives a {\em denial message} from node agent
$v_\facilitator^i$ for reserving $v_\Next^i$, the message contains the
{\em suggestion of movement} (SOM), {\em \waitSOM} (i.e., waiting for
a while) or {\em \detourSOM} with another next node in $V_\main$
(i.e., taking a detour) which neighbors $v_\current^i(t)$. When the
SOM is \detourSOM, $i$ leaves for the specified next node and
generates another (shortest) path $p_i$ from that node to the
destination using a conventional algorithm without considering the
planned paths of other agents. Note that it is probable that the
generated new path contains $v_\Next^i$ that was denied; however, all
edges in $G_\main$ are directed, thus $i$ should take a detour to
return to $v_\Next^i$, therefore, its surrounding situation becomes different. 
\par

\begin{figure}
  \centering
  \includegraphics[width=0.85\hsize]{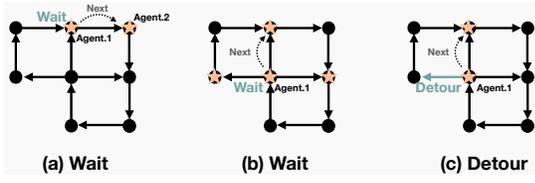}
  \caption{Example of the denial message with SOM.}
  \label{fig:waitDetour}
\end{figure}

\subsection{Node Agent Behavior for Conflict Detection}
Node agent $v \in V_\main$ manages the reservation of the
corresponding node $v$ for the {\em staying (carrier) agent} $i$ at
$v$ as the facilitator. It confirms whether $i$ can move to the next
node by communicating with the neighboring node to determine the
possibility of a collision. This implies that more than two agents
$i,j\in\AgentSet$ attempt to reserve the same node simultaneously, but
we assume that the node agent reads them from its message queue
  one by one.

\par

When node agent $v$ receives a request message to move to neighboring
node $v_\Next^i\in V_\main$ from agent $i$ on $v=v_\current^i(t)$, $v$
asks the vacancy to node agent $v_\Next^i$ by sending a reservation
message. If $v_\Next^i$ is reserved by no other agent at $t+1$,
$v_\Next^i$ reserves its resource for $i$ and $v_\Next^i$ sends the
acceptance message $v$ and it is forwarded to $i$. If $v_\Next^i$ is
already reserved at $t+1$ or another agent $j$ on $v_\Next^i$ is not
decided to move to $v_\Next^j$ at $t+1$, facilitator $v$ forwards a
denial message from $v_\Next^i$ to $i$ with the possible action
label, which is among the following SOMs:
\begin{description}
\item[\WaitSOM:] Node agent $v$ suggests for $i$ to extend the current stay until $t+1$ (thus the extension is not necessarily {$T_\move$}). This is always possible because $v$ accepts the reservation request from another agent only after $i$ has reserved the next node.
\item[\DetourSOM:] Suppose that $v$ has multiple outward-direct
  edges. $v$ sends the reservation message to the neighboring node,
  consecutively, except $v_\Next^i$, and if one of them
  $\tilde{v}_\Next^i$ accepts it, $v$ sends the denial message with
  \detourSOM and the reserved node $\tilde{v}_\Next^i$ for $i$.
\end{description}

Node agent $v$ should select which SOM, \waitSOM or \detourSOM,
depending on the situation. Both SOMs have their pros and cons; the
\waitSOM SOM may block other agents, whereas the \detourSOM SOM may
force $i$ to take a detour. However, in our experiments, node agent
$v$ attempted to send the \detourSOM, and when it was not possible, it
sent \waitSOM, because ensuring that the agent's flow is not disrupted
is effective in the efficient execution of tasks, particularly in a
crowded situation. Examples are shown in Fig.~\ref{fig:waitDetour};
the facilitator of Agent~1 sends a denial message with \waitSOM in
Figs.~\ref{fig:waitDetour}a and b because all neighboring nodes are
already reserved, whereas the facilitator sends the message with
\detourSOM in Fig.~\ref{fig:waitDetour}c because two neighboring nodes
are not reserved and one of them is randomly selected, although the
next node planned in Agent~1 has been reserved. 
\par

\subsection{Collision Detection in the Marginal Zone}
First, suppose that tree-structured subgraph (area) $G_\tree^k$ does not include parking nodes. If node agent $v\in V_\Root$ is the root of $G_\tree^k$, $v$ also manages to restrict the number of agents entering $V_\tree^k\setminus \{v\}$ to one. Therefore, if agent $i$ on $v$ at time $t$ (i.e., $v=v^i_\current(t)=v^i_\facilitator$) attempts to move to $v_\Next^i\in V_\tree^k$, $i$ sends a request message to $v$. Then, $v$ sends back an acceptance message to $i$ only when no other agent is currently in $G_\tree^k\cup G_\marginal$; otherwise, $v$ sends a denial message with detour as a SOM if possible. Moreover, if $v$ cannot find the neighboring node to which $i$ can move, $v$ sends the denial message with \waitSOM.
\par

Meanwhile, when $v_\Next^i=v$ and $v_\current^i(t)\in V_\tree^k\cap V_\marginal$, $i$ sends the request to its facilitator agent $v$ ($=v_\facilitator^i=v_\Next^i$) to reserve $v_\Next^i$. If $v$ can reserve itself for $i$ at $t+1$, $v$ sends an acceptance message to $i$, otherwise, it sends a denial message with \waitSOM. Agent $i$ on $ v^i_\current \in G_\tree^i\cap G_\marginal$ does not send a request message if $v_\Next^i\not\in V_\main$.
\par

When $G_\tree^k$ includes several parking nodes, its root node $v\in
V_\Root$ has two techniques to manage the number of agents entering
it. One technique is to make it one-way as the MAPDFS progresses; that
is, at the beginning of a MAPDFS instance, $v$ restricts the direction
of movement only to the main area and thereafter only to the interior
of the tree area. However, once an agent returns to the parking node,
it cannot go back to the main area. Another method is by managing
agents entering the tree area with \remove{task} endpoints; $v$ restricts the
number of agents entering $V_\tree^k\setminus \{v\}$ to one; however,
when it arrives at the parking node, $v$ ignores it; thus, another
agent can enter this area. Meanwhile, when agent $i$ at the parking
node attempts to go to the main area, it asks its facilitator
$v_\facilitator^i=v$ for the possibility of leaving. Then, $v$ accepts
it only when there are no other moving agents in $G_\tree^k$;
otherwise, $v$ sends a denial message with \waitSOM as a SOM to
$i$. We used the first technique in our experiments below.

\subsection{Number of Open Nodes in the Main Area}\label{sec:open}
Finally, we discuss the efficiency and difference between the numbers
of agents and nodes in the main area. We call a node that is not
reserved by any agent as an {\em open node}, and any agents can move
to the next nodes only when they are open. Because each agent reserves
one different node, if $\abs{V_\main}=\abs{\AgentSet}$, then all
agents cannot move anyhow. If $\abs{V_\main}-\abs{\AgentSet}=1$ and
all agents are in $G_\main$, agents are confined within the current
bi-connected components and cannot move to neighboring bi-connected
components. This is because, when agent $i$ in $G_k$ successfully
reserves the next open node, the current node becomes open in the next
time; thus, an open node looks like moving backward. Therefore, for
agent $i$ to enter to another neighboring bi-connected component
$G_{k'}\setminus G_k$, an open node should also be in $G_{k'}\setminus
G_k$ and in front of $i$. However, this situation cannot happen if
there is only one open node.
\par

When $\abs{V_\main}-\abs{\AgentSet}=2$, it is possible that agent $i$
enters to a neighboring bi-connected component $G_{k'}\setminus G_k$
only when two open nodes are in $G_k\setminus G_{k'}$ and
$G_{k'}\setminus G_k$. However, such a situation can happen but is
mostly coincidental. Therefore, agents can reach their destination but
almost randomly; agents can complete all tasks but it will take a long
time. Consequently, it is evident that the more open nodes
are in the main area, the easier it is for agents to move to the
desired neighboring node. Thus,\remove{Therefore,} it is recommended
that more than half of the main area have open nodes for efficiency.
\par

\def\SpeedDelayProb{\nu}
\def\abs#1{\vert #1 \vert}
\def\noise{{\it nse}}

\begin{table}
 \caption{Experimental parameter values.}
 \label{table:expTable}
 \centering
 \tabcolsep=5pt
 \begin{tabular}{llll}
   \toprule 
   Parameter description and symbol & Value \\
   \midrule 
   Normal required time for neighboring node, $T_\move$  & 3\\
   Time to load/unload, $T_\loadunload$ & $3$ or $6$\\   
   Moving delay probability, $\SpeedDelayProb$ & 0 to 0.2 \\
   Moving delay time, $T_\noise$ & 1 or 2 \\ \bottomrule 
  \end{tabular}
\end{table}


\begin{table*}
  \caption{Completion rate and planning time of MAPDFS instances.}
  \label{table:completedTask}
  \centering
  \tabcolsep=3pt 
  \begin{tabular}{cccccccccccccccccccc}
    \toprule
    \multicolumn{3}{c}{} & \multicolumn{14}{c}{Number of agents} & \\ 
    \cmidrule{3-19}
    \multicolumn{1}{c}{} & \multicolumn{1}{c}{Alg.}
    & 2 & 4& 6 &8 & 10 & 12 & 14 & 16 & 18 & 20 & 22 & 24 & 26 & 28 & 30 & 35 & 40 \\
    \midrule
    completion rate &Proposed $\SpeedDelayProb=0$&1.0&1.0&1.0&1.0&1.0&1.0&1.0&1.0 &1.0&1.0&1.0 &1.0&1.0&1.0&1.0&1.0&1.0 \\
       		     &HTE &1.0&1.0&1.0&1.0&1.0&1.0&1.0&1.0 &1.0&1.0&1.0 &1.0&1.0&1.0&1.0&1.0&1.0 \\
       		     &RHCR &1.0&0.92&0.86&0.78&0.52&0.28&0.18 &0.02&0.0&0.0 &0.0&0.0&0.0&0.0&0.0&0.0&0.0 \\
    \midrule
    planning time  &Proposed $\SpeedDelayProb=0$&0.09&0.15&0.22&0.29&0.36&0.46&0.54&0.63 &0.71&0.84&0.96 &1.09&1.20&1.36&1.56&2.09&2.78 \\
       	         	  &HTE &11.2&11.6&12.2&13.6&15.0&15.1&14.9&15.0 &15.1&15.1& 15.1&15.5 &15.4&15.1&15.3&15.5&15.5 \\
       		          &RHCR &90.8&128.1&167.5&209.1&270.4&358.7&421.9&640.4&-&- &-&-&-&-&-&-&- \\    \bottomrule
  \end{tabular}
\end{table*}

\section{Experiments and Discussion}

\subsection{Experimental Setting}
We evaluated the proposed method using MAPDFS instances in two
environments that satisfy our required conditions (SC1- SC3) and are
likely to appear in our application (Fig.~\ref{fig:env}). They have a
small number of task endpoints (blue dots) that correspond to
load and unload nodes for tasks. The first environment (Env.~1)
in Fig.~\ref{fig:env1} has 10 task endpoints which are placed at the
ends of tree-structured areas and satisfy the WFI condition for other
methods (e.g., HTE). The second environment (Env.2)
in Fig.~\ref{fig:env2} also has 10 task endpoints, but they are placed
in the main area; thus, it does not satisfy the WFI condition. When an
agent loads or unloads at the task endpoint, it may block other agents
for a while until the loading/unloading is completed; however, we
think that this is a common practice in construction sites and an
inevitable part of the process. 
\par

For comparison, we implemented two existing methods as baseline, {\em HTE~\cite{MaAAMAS2017} and {\em RHCR~\cite{Li2021AAAI}. HTE is a decentralized method in which agents plan their paths consecutively by referring to the synchronized shared memory block that contains information about the task set and all agents' paths, including visit and leave times. RHCR decomposes a MAPD problem into a sequence of windowed MAPF instances. Then, agents in RHCR find and resolve conflicts that occur within the next $w$ timesteps and replan paths every $h$ timesteps. Although RHCR does not require the WFI condition unlike HTE, the deadlock avoidance in RHCR is incomplete. Note that we used the {\em priority-based search}~\cite{Ma2019AAAI} for MAPF solver of RHCR, and set $w=60$ and $h=15$ by referencing the original experiments~\cite{Li2021AAAI}, in which $w=20$ and $h=5$; thus, we multiplied them by $T_\move=3$.
\par

In the first experiment (Exp.~1), we compared the performance of our method with those of the two baseline methods in Env.~1, which satisfied the WFI condition for HTE with constant moving speed, because the baseline methods cannot handle delays. In the second experiment (Exp.~2), we confirmed whether agents in our method could complete all tasks in Env.~2, which did not meet WFI condition and had a small negative swing (i.e., delay). Therefore, we conducted Exp.~2 only using our method. We set $T_\loadunload=6$ in Exp.~2, whereas $T_\loadunload=3$ in Exp.~1. Thus, agents blocked other agents longer. We used three evaluation measures: (1) the rate of completion of MAPD instances, (2) makespan (i.e., the time required to complete all tasks), and (3) planning (CPU) time. 
\par

All agents started from their parking nodes (green dots in Fig.~\ref{fig:env}). The loading and unloading nodes for a task were randomly selected from the set of task endpoints (blue dots), and the initial tasks were assigned simultaneously to all agents. An agent was assigned a new task after completing the current task, and this process was repeated until all tasks in $\TaskSet$ were completed. If an agent was not assigned a task because all tasks had already been assigned, it returned to its parking node.
\par

To model more realistic robotic movements, we added noise to the moving speed $T_\noise$ at probability $\SpeedDelayProb$ ($0\leq \SpeedDelayProb \leq 1$) when an agent moved to the next node in Exp.~2. Therefore, for agents to often move to the neighboring node in $T_\move=3$ timesteps, but with probability $\SpeedDelayProb$, it required $T_\move + T_\noise$ timesteps, where $T_\noise$ was randomly selected by either $1$ or $2$. We set the number of agents from 2 to 40 and the number of tasks to $\abs{\TaskSet}=100$. We have listed all other parameter values in Table.~\ref{table:expTable}. All experimental data are as the average of $50$ independent trials using apple M1 Max CPU with 64 GB RAM.

\subsection{Performance Comparison}
\subsubsection{Completion Rate}
We investigated the rate of completeness in our 50
runs. Table~\ref{table:completedTask} lists the rate of completed
MAPDFS instances \remove{of the MAPDFS instances} with different
numbers of agents in 
Exp.~1. Here, we considered it as a failure if running time exceeded
the limit of timestep ($10000$ timestep), the planner could not find a
collision-free path, or a collision
occurred. Table~\ref{table:completedTask} indicates that agents in our
method could complete all instances in Exp.1 without failures and
collisions. HTE could also complete all instances in Exp.~1 because
Env.~1 meets the WFI condition. However, the completion rate of RHCR
rapidly decreased with the increasing number of agents $n$ and became
zero eventually when $n\geq 18$. This is because agents often headed
for the few same task endpoints as their destinations and the areas
near the task endpoints were congested. Thus, the prioritized planning
in RHCR seemed difficult to avoid collisions in these situations. Note
that the data are not listed here, but we also conducted Exp.~2 using
the baseline methods but they failed in all instances, although our
method completed all instances. 
\par

\begin{figure}
  \centering
  \includegraphics[width=0.7\hsize]{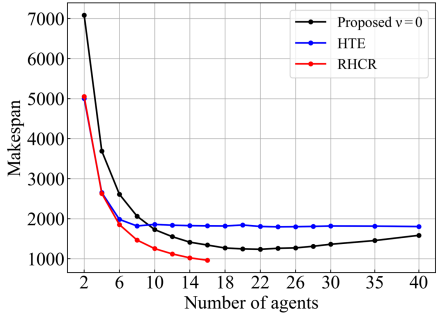}
  \caption{Makespans in Exp.~1.}
  \label{fig:makespan_exp1}
\end{figure}

\begin{figure}
  \centering
  \includegraphics[width=0.7\hsize]{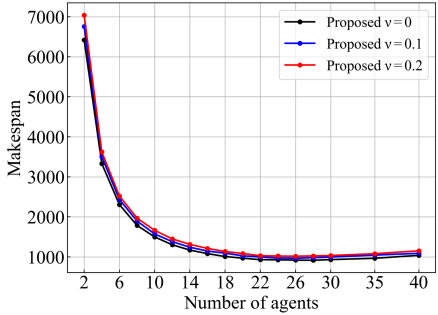}
  \caption{Makespans in Exp.~2.}
  \label{fig:makespan_exp2}
\end{figure}

\begin{figure}
  \centering
  \includegraphics[width=0.7\hsize]{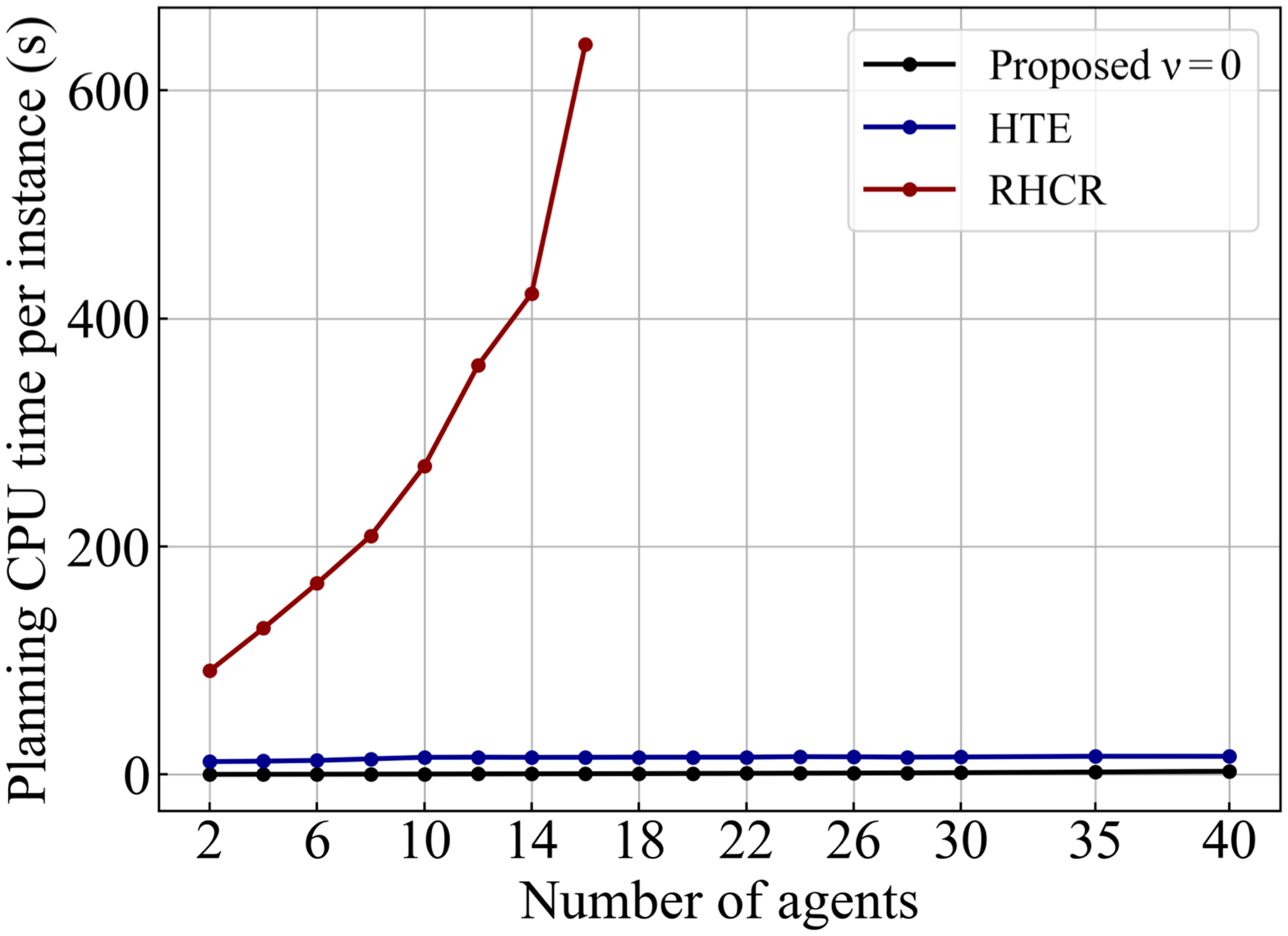}
  \caption{Planning time per instance (s) in Exp.~1.}
  \label{fig:cputime_exp1}
\end{figure}

\begin{figure}
  \centering
  \includegraphics[width=0.7\hsize]{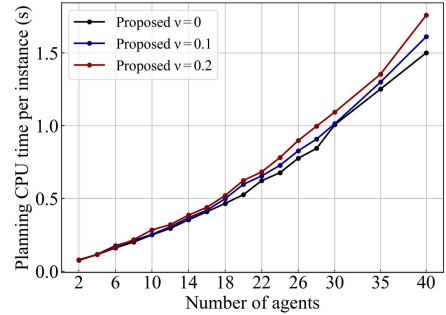}
  \caption{Planning time per instance (s) in Exp.~2.}
  \label{fig:cputime_exp2}
\end{figure}

\subsubsection{Makespans}
Figure~\ref{fig:makespan_exp1} plots the average makespan (in timesteps) with different numbers of agents in Exp.~1. Note that the failure instances were excluded from the average calculation. This figure shows that, even if the number of agents $n$ increased to approximately six, the makespan could be shortened regardless of the methods employed.
\par

However, when $n\geq 8$, we observed performance differences with
these methods; agents with the proposed method exhibited the best
performance. They could gradually decrease their makespans with an
increasing number of agents up to $22$, but the performance slightly
degraded when $n\geq 24$. This small degradation was caused by
over-crowded areas near task endpoints by increasing the number of
agents. Meanwhile, when $2\leq n\leq 6$, the performances of the
baseline methods were better than that of the proposed method. Even if
the number of task endpoints was not significantly large, the baseline
methods enabled the agent to move the environment in
parallel. However, agents with the proposed method were sometimes
forced to take longer detours by following the directions of edges. 
\par

Conversely, the performance of the agents with HTE was almost constant
when, $n\geq 8$. Whereas Env.~1 met the WFI conditions, it had a
small number of task endpoints that were fewer than the agents. Thus,
HTE could assign a limited number of tasks to agents because an agent
could not select or be assigned the task whose loading or unloading
node was already reserved as the task endpoints of other being
executed tasks. 
\par

Although RHCR considerably outperformed other methods until $n\leq 16$ in Exp.~1, after the number exceed $16$, we could not calculate the makespan because no instances of the MAPDFS problem could be completed by RHCR owing to the congestion as discussed before. Even when $n=12$, the completion rate by agents with RHCR was\remove{approximately} $0.28$ from Table~\ref{table:completedTask}; it is not realistic to use RHCR in our target applications because of the low completion rate. 
\par

We plotted the average makespan for Exp.~2 in Fig.~\ref{fig:makespan_exp2} to investigate the effects of the fluctuation on movement speed on the makespan in MAPDFS. This figure shows that the performance gradually decreased with increasing moving delay probability $\SpeedDelayProb$. However, their difference is insignificant. Therefore, the result indicates that the proposed methods are robust against the fluctuation in movement speed. We also conducted the experiments by setting $T_\loadunload=3$; however, the difference was significantly small. We believe that this effect was caused by the orientation in the main area; although an agent had to wait for loading/unloading of other agents, they started to move in the same direction according to the orientation. Thus, agents did not have to worry about head-on collisions and could wait next to where it was loading/unloading.

\subsubsection{Planning (CPU) time}
Figure~\ref{fig:cputime_exp1} shows the averaged planning time for all agents per instance with a different number of agents in Exp.~1. We have also listed the detail of the total planning time in Table~\ref{table:completedTask}. Clearly, the planning time with the proposed method is much smaller than those of other methods, regardless of the number of agents. This is because, unlike HTE, agents with the proposed method could generate paths without time information and without considering other agents' paths.
\par

Clearly from the figure, the planning time with HTE was almost
identical when the number of agents was $n\geq 10$ because the number
of agents moving in parallel was limited and only the active agents
generated plans. Conversely, in the proposed method, the planning time
was slightly increased based on the increase in $n$, because all
agents move in parallel and require time for their planning. However,
even when $n=40$, the total planning time of the proposed method was
only 2.78 seconds owing to the simple distributed planning. Meanwhile,
RHCR required considerably large planning time
(Fig.~\ref{fig:cputime_exp1}) because all agents are required to
replan at least once in $h=15$ timesteps by interleaving planning and
execution.
\par

Figure~\ref{fig:cputime_exp2} shows the averaged planning time per instance with the moving delay probability $\SpeedDelayProb=0, 0.1, 0.2$ in Exp.~2. Unlike Exp.1, agents often were forced to stay longer at the same nodes in the main area due to delay by other agents and loading/unloading actions of other agents. Thus, agents might receive denial messages more frequently. However, the effect of fluctuation in movement speed on planning time was significantly small from this figure.

\subsection{Discussion}
The proposed method completed all tasks without collision or entering
deadlock states using distributed planning with asynchronous execution
in the\remove{our target} environments that satisfy our required conditions, such
as the robustness to fluctuated speed. This is because node agents
prevented carrier agents from moving to the next nodes that are
already reserved and/or stayed by other agents, regardless of delay
owing to the speed variations, and direct edges prevented agents from
crossing the same edge in opposite directions.
\par

Furthermore, our method outperformed the baseline methods (HTE and
RHCR) in environments with a small number of endpoints when the number
of agents was more than 10. Our experimental results show that the
proposed method can increase the concurrency of task execution and
mitigate the performance degradation caused by crowded regions. HTE
was unable to increase the number of concurrent task executions
because the number of task endpoints was smaller than the number of
agents. Further, RHCR required to replan repeatedly in all agents with
synchronous planning and execution, which increased computational cost
and could not complete all tasks within a reasonable time because of
many live and deadlock situations caused by congestion. 
\par

In the proposed method, after an agent generated a path to the
destination, it moved to the next node with local communications to
check the availability and modify the path if necessary. Therefore, it
was considered similar to the family of {\em local repair algorithms}
with limited window size, such as traditional {\em local repair
  A*}~\cite{LRA:1992}, its extension algorithms, and RHCR. A drawback
of this type of algorithm is that, when many agents gather at a small
number of nodes, they may cause a high likelihood of collisions, many
livelock states, and many costly repairs/replanning because of
congestion. However, as the proposed method introduces an orientation
into the graph, it prevents, for example, an agent from being
sandwiched between other agents coming from the left and
right. Moreover, because agents are navigated in the direction in
which they can move, even when agents' destinations are concentrated
and crowded, agents can be moved temporarily to surrounding areas to
maintain their mobility. 
\par

Finally, we have to discuss more on the number of open nodes in the
main area and the number of (carrier) agents. As mentioned in
Section~\ref{sec:open}, when $\abs{V_\main}-\abs{\AgentSet}=2$, only
two agents can start moving simultaneously but almost randomly. This
restriction is significantly different in decentralized methods
assuming synchronous
movements~\cite{MaAAMAS2017,ma2019lifelong,WiktorIROS2014,Okumura2019IJCAI},
in which all agents move synchronously; thus, agents can move even
when $\abs{V_\main}=\abs{\AgentSet}$. However, we assumed the
asynchronous movements, and such movements are impossible. Furthermore,
because of the asynchrony in the distributed environments, how to move
is affected by the timing of activities, such as the time/order of
message arrivals. Thus, the performance is partially affected by
randomness. However, in our extensive experiments using our method,
agents completed all tasks.

\section{Conclusion}
We presented a distributed planning with asynchronous execution
methods which is an efficient and robust solution for realistic
environments. Our method is simple yet applicable to environments that
have a smaller number of task endpoints than agents and include
  the fluctuated movement speed of agents. From our experiments, the proposed
method outperformed baseline methods for MAPD problem and even in the
environments to which they are not applicable because of variable
speed and flexible endpoint locations. Our method completed all tasks
efficiently without collision and deadlock in such environments. 
\par

In the future, we plan to extend our method; for example, we will
relax environmental graph conditions, propose appropriate graph
orienting to improve the effectiveness and efficiency, and address
complex tasks (e.g., a task can be executed by multiple agents).






\bibliographystyle{ACM-Reference-Format} 
\bibliography{bib_desk}

\clearpage

\section*{Appendix}
\subsection*{Proposition 1 with Detailed Proof}
We will provide the detailed proof of Proposition 1, although it is
intuitively obvious.

\addtocounter{proposition}{-1}

\begin{proposition}
Suppose that $G_\main=G_1\cup \dots\cup G_K$ is connected graph and $G_k$ is the
bi-connected component.
If $V_k\cap V_{k'}\not=\varnothing$, the $V_k\cap V_{k'}$ is a singleton.
\end{proposition}

\def\tm{{\widetilde{m}}}

\begin{proof}
If $V_k\cap V_{k'}$ contains two nodes, $w$ and $w'$, we can show that
$V_k\cup V_{k'}$ is bi-connected. Therefore, we demonstrate that, for
any nodes $v\in V_k$ and $v'\in V_{k'}$, there exist two paths in
$G_k$ between $v$ and $w$ and between $v$ and $w'$, whose common nodes
are only $v$. Thus, we can also generate two paths in $G_{k'}$ between
$v'$ and $w$ and between $v'$ and $w'$, whose common nodes are only
$v'$. Then, by connecting these paths from $v$ to $v'$ through $w$ and
from $v$ to $v'$ through $w'$, we can proof that $V_k\cup V_{k'}$ is
bi-connected. First, because $G_k$ is bi-connected, we can generate at
least two paths in $G_k$ connecting between $v$ and $w$ such that
these paths do not have the common node, except $v$ and $w$.
We denote these paths by sequences of nodes, $p_{v\To
  w}^0=(v_0^0,v_1^0, \dots, v_{h_1}^0)$, and $p_{v\To w}^1=(v_0^1,
\dots, v_{h_2}^1)$, where $v_0^0=v_0^1=v$, $v_{h_1}^0=v_{h_2}^1=w$,
and $v_l^0$ and $v_{l'}^1\in V_k$. Similarly, we can generate two
paths between $v$ and $w'$ whose common nodes are only end
nodes. These paths are denoted by $p_{v\To w'}^0=({v'}_0^0, \dots,
{v'}_{{h'}_1}^0)$ and $p_{v\To w'}^1=({v'}_0^1, \dots,
{v'}_{{h'}_2}^1)$, where ${v'}_0^0={v'}_0^1=v$,
${v'}_{h_1}^0={v'}_{h_2}^1=w$, and ${v'}_l^0$ and ${v'}_{l'}^1\in
V_k$.
If $\exists m\in\{0,1\}$ such that $p_{v\To w}^m$ that intersects,
except at their end nodes, at most one of $p_{v\To w'}^1$ and
$p_{v\To w'}^2$, we can select $p_{v\To w'}^{m'}$ that does not
intersect $p_{v\To w}^m$.
If such a path does not exist, $\forall m\in\{0,1\}$ such that
$p_{v\To w}^m$ that intersects both $p_{v\To w'}^0$ and $p_{v\To
  w'}^1$. Let $v_l^m$ be the last nodes in $(v_0^m,v_1^m, \dots,
v_{h_m}^m)$ at which intersect $p_{v\To w'}^0$ or $p_{v\To
  w'}^1$. Assuming that $v_l^m$ is selected from $p_{v\To w'}^\tm$
($\tm\in\{0,1\}$), if
we consider paths $p_{v\To w}=({v'}_0^\tm, \dots, {v'}_{l'}^\tm=v_l^m,
\dots, v_h^m, v_{h_m}^m)$ and $p_{v\To w'}^{1-\tm}$, these have no
intersection except $v$.
\end{proof}

\begin{figure}
  \centering
  \includegraphics[width=0.85\hsize]{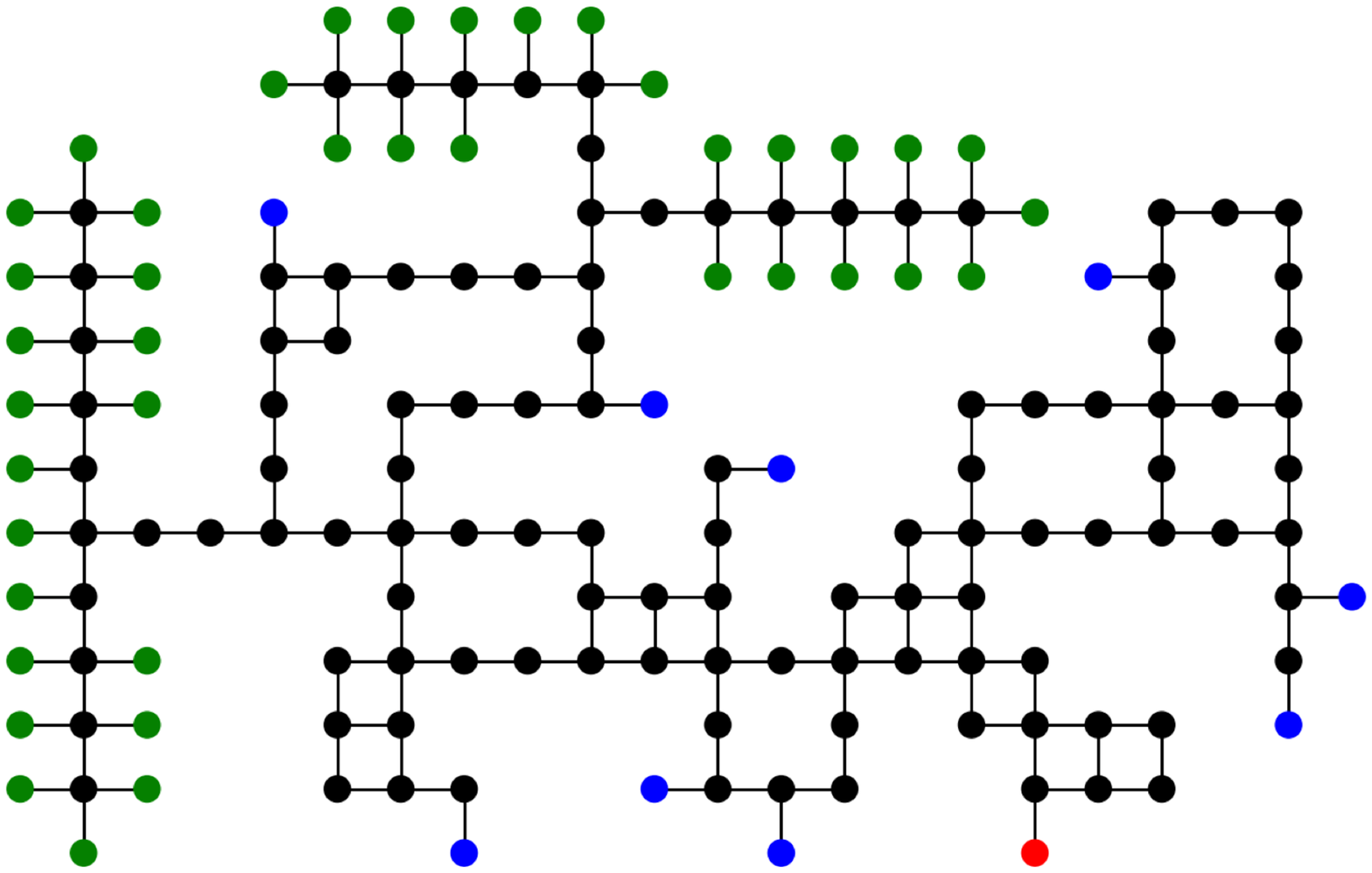}
  \caption{Example environments (Env.3).}
  \label{fig:env_3}
\end{figure}

\begin{figure}
  \centering
  \includegraphics[width=0.85\hsize]{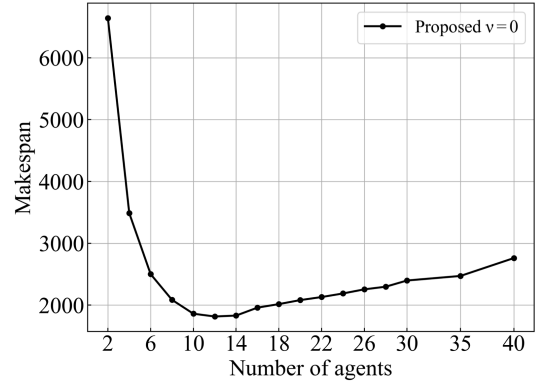}
  \caption{Makespans in Exp. 3.}
  \label{fig:makespan_exp3}
\end{figure}

\begin{figure}
  \centering
  \includegraphics[width=0.35\hsize]{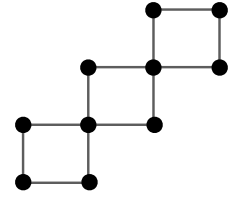}
  \caption{Example environments (Env.4).}
  \label{fig:env_4}
\end{figure}

\subsection*{Videos of Agents' Movements in (Additional) Experiments}
We have added videos of the agents' movements in our
experiments in this appendix.
We also conducted two additional experiments to confirm whether agents
using the proposed method can complete all task without collision in
rare environments where over-crowded often occur. We briefly describe
these experiments below.
\par

In the first additional experiment (Exp.3), the structure of environment
was identical to Exp.1, however, the loading node was selected from
only red node, and the unloading nodes were selected from blue node as
shown in Fig.~\ref{fig:env_3}.  This assumes that the materials are
stored at the specific place (red node) and are carried to various nodes;
therefore, over-crowded situations are likely to occur around the red node. 
\par

Figure~\ref{fig:makespan_exp3} plots the average makespan with
different number of agents $n$ in Exp.3. Note that we conducted fifty
experimental runs as in other experiments.
The experimental result shows that agents using proposed method could
also complete all instance in Exp.3 regardless of the number of agents
whose max is $40$.
The agents with proposed
method could reduce the makespan along with the increasing of the number of
agents to $n=12$. After that, as the number of agents increased when
$n\geq 14$, the performance gradually decreased probably due to
over-crowded situations around a single loading node (red node).
Furthermore,  we can see that 
makespan was considerably larger in Exp.~3 by comparing
Figs.~6, 7 and \ref{fig:makespan_exp3}. This is also because there is
only one endpoint for loading, and thus, all agents are likely to
gather this endpoint. We want to insist that agents with the proposed
method still perform MAPD instances without deadlock
situations as shown in the following video, indicating
that our proposed methods are robust even in such a over-crowded
situation.
\par

In the second additional experiment (Exp.4), we tested the relevance
between the number of nodes in main area and the number of agents,
mentioned in main paper. There are eight agents running the MAPD
problem instance in three bi-connected components whose the total
number of nodes is ten, as shown in Fig.~\ref{fig:env_4}; therefore,
$\abs{V_\main}-\abs{\AgentSet}=2$. We also assume that all nodes are
endpoints, so an agent can be load and unload any node.
All agents initially placed at the random nodes and began to perform
the pickup-and-delivery tasks. We believe that this situation was
pathological. However, we confirmed that agents could
complete an MAPD instance in Exp.4, although agents could reach their
destination almost randomly so required so long time.
\par

The followings are list of the sample videos to show the movements of
the agents using proposed method in our four experiments.
Note that, due to limit of submission file size,  we make the short
clip of video. Note that in these videos, the direction of an edge is
represented by a long triangle, whereas the undirected edge is represented by a bold line.

\subsubsection*{\bf Experiment 1 (Exp.1)}\mbox{}\\
Agents used proposed method, the number of agent was 22, the
environment is Env.~1, and
sample video is available at the url \url{https://youtu.be/l4xnnsy5TJs}.
We can  confirm that all agents were able to move around the area
without collisions, although light congestion occasionally occurred in
some places because several agents set the same destinations.

\subsubsection*{\bf Experiment 2 (Exp.2)}\mbox{}\\
Agents used the proposed method, the number of agent was 40, the
environment is Env.2. Sample video can be found at the url \url{https://youtu.be/do3pa22yKps}.
Although Env.~2 which does not meet the WFI condition and has a small
negative swing, they could smoothly move around there and complete the
MAPD instances.

\subsubsection*{\bf Experiment 3 (Exp.3)}\mbox{}\\
Forty agents adopting the proposed method performed the
MAPD instance in the environment Env.~3. Sample video can be found at \url{https://youtu.be/uXfgFJjgLIA}.
As we mentioned above, agents using proposed method could complete all
instance in such a over-crowded situation.

\subsubsection*{\bf Experiment 4 (Exp.4)}\mbox{}\\
Eight agents adopting the proposed method performed the
MAPD instance in the environment Env.~3 that has only ten nodes.
TSample video can be found at \url{https://youtu.be/0ap6Vq9JbBw}.
As we mentioned above, agents using proposed method could complete all
instance when $\abs{V_\main}-\abs{\AgentSet}=2$.


\end{document}